\documentclass[3p, times]{elsarticle}
\makeatletter
\usepackage{setspace}
\usepackage[numbers]{natbib}
\usepackage{notoccite}
\usepackage{epsfig}
\usepackage{rotating}
\usepackage{amssymb}
\usepackage{xcolor}
\usepackage[T1]{fontenc}
\usepackage{lmodern}
\usepackage[utf8]{inputenc}
\newenvironment{proof}{\paragraph{Proof:}}{\hfill(Proved)}
\usepackage{booktabs}
\usepackage{multirow}
\usepackage{enumerate}
\usepackage{caption}
\usepackage{subcaption}
\usepackage{lscape,float}
\usepackage{graphicx}
\usepackage[superscript,biblabel]{cite}
\usepackage{xcolor}
\usepackage{soul}

\usepackage{natbib,hyperref,doi}
\usepackage[superscript,biblabel]{cite}
\makeatletter
\renewcommand\@biblabel[1]{#1.}
\makeatother
\usepackage{algorithm}
\usepackage[noend]{algpseudocode}

\usepackage[fleqn]{amsmath}

\usepackage{mathtools}
\newtheorem{theorem}{Theorem}
\newtheorem{result}{Result}

\makeatletter
\renewcommand\@biblabel[1]{#1.}
\makeatother
\usepackage{longtable}

\begin{document}
%\baselineskip=24pt
%\singlespacing

%\parskip = 10pt
%\def \qed {\hfill \vrule height7pt width 5pt depth 0pt}

%\def\refhg{\hangindent=20pt\hangafter=1}
%\def\refmark{\par\vskip 2mm\noindent\refhg}
%\def\refmark{\par\vskip 2.50mm\noindent\refhg}
%\include{dbtweaged}
%\include{chirp_rev1recent.tex}
%\include{dbtlnagad_blind}
%\include{tksdk1}
%\include{realam_csa}
%\usepackage{amsmath}

\begin{frontmatter}

\title{Bayesian reliability acceptance sampling plans for competing risks data under interval censoring} %\tnoteref{label0}}
%\tnotetext[label0]{This is only an example}

\author[label1]{Biswabrata Pradhan \corref{cor1}}
%\address[label5]{Some University}

\address[label1]{Statistical Quality Control \& Operations Research Unit, Indian Statistical Institute, Kolkata, India}
%\address[label2]{Address Two\fnref{label4}}

\cortext[cor1]{Corresponding author}

\journal{Springer Handbook of Reliability}
\ead{bis@isical.ac.in}
%\ead[url]{author-one-homepage.com}
\author[label1]{Rathin Das}

\begin{abstract}
%\doublespacing
We obtain a reliability acceptance sampling plan for independent competing risk data under interval censoring schemes using the Bayesian approach. At first, the Bayesian reliability acceptance sampling plan is obtained where the decision criteria of accepting a lot is pre-fixed.  For large samples, computing Bayes risk is computationally intensive. Therefore, an approximate Bayes risk is obtained using the asymptotic properties of the maximum likelihood estimators. Lastly, the Bayesian reliability acceptance sampling plan is obtained, where the decision function is arbitrary. The manufacturer can derive an optimal decision function by minimizing the Bayes risk among all decision functions. This optimal decision function is known as Bayes decision function. The optimal sampling plan is obtained by minimizing the Bayes risk. The algorithms are provided for the computation of optimum Bayesian reliability acceptance sampling plan. Numerical results are provided and comparisons between the Bayesian reliability acceptance sampling plans are carried out.  

\end{abstract}
\end{frontmatter}
%\tableofcontents
\addcontentsline{toc}{section}{Abbreviation \& Notaion}
\section*{Abbreviation \& Notaion}
 \begin{longtable}{ll}
 ASP&acceptance sampling plan\\
 RASP & reliability acceptance sampling plan\\
PDF & probability density function\\
RF& reliability function\\
RV&random variable\\
CDF& cumulative distribution function\\
FIM& Fisher information matrix\\
JPDF& joint probability density function\\
MLE& maximum likelihood estimator\\
BSP& Bayesian RASP\\
IID& independent and identically distributed\\
wrt&with respect to\\
$X_i$& RV of lifetime of the $i^{th}$ item\\
$\boldsymbol{\nu}$ & vector of parameters\\
$f(\cdot\ | \ {\boldsymbol{\nu}})$& PDF of $X_i$\\
$C_i$& cause of the failure of the $i^{th}$ item\\
$G(\cdot\ | \ {\boldsymbol{\nu}})$& sub-distribution of $X_i$ and $C_i$\\
$F(\cdot \ | \ {\boldsymbol{\nu}})$& CDF of $X_i$\\
$n$ & sample size\\
$k$& inspection number\\
$\tau_1,\tau_2,\ldots,\tau_k$& inspection times points\\
$\boldsymbol{\tau}=(\tau_1,\tau_2,\ldots,\tau_k)$ & vector of inspection time points\\
$\boldsymbol{\zeta}=(n,\boldsymbol{\tau},k)$& vector of sampling parameters\\
$D_{mj}$& RV denoting due to cause $j$, the number of failures in the\\
&$(\tau_{m-1},\tau_m]$ interval\\
$\boldsymbol{D}=(D_{11},\ldots,D_{1J},\ldots,D_{m1},\ldots,D_{mJ})$& vector of number of failures in $m$ intervals due to each cause\\
$d_{mj}$& observed value of $D_{mj}$\\
$\boldsymbol{d}=(d_{11},\ldots,d_{1J},\ldots,d_{m1},\ldots,d_{mJ})$& vector of observed number of failures in $m$ intervals\\
&due to each cause\\
$P_{\boldsymbol{D}}$& joint distribution of $\boldsymbol{D}$\\
$D_t=\sum_{m=1}^k\sum_{j=1}^JD_{mj}$& RV denoting the total number of failures after life test\\
$d_t$& observed value of $D_t$\\
$E[D_t]$& expected number of failures \\
$E[M]$& expected inspection number\\
$E[\tau]$& expected test duration\\
$E[D_t^*]$& expected number of failures for optimal BSP \\
$E[\tau^*]$& expected test duration for optimal BSP\\
$E[M^*]$& expected inspection number for optimal BSP\\
$P(A)$&probability of accepting the lot for BSP\\
\end{longtable}
\section{Introduction}
Acceptance sampling plan (ASP) is a popular method for deciding on acceptance or rejection of a lot of products. RASP is a special type of ASP where the decision is taken based on product's reliability. In this method, a sample is drawn for life testing from the lot. After life testing, using the lifetime data, a suitable statistic is obtained to take a decision of acceptance or rejection about the lot. In reliability studies, determining an optimal RASP is a crucial task. The methods for choosing an optimal RASP are defense sampling schemes, producer's and consumer's risk method, decision-theoretic approach and Dodge and Roming's plan. From an economic perspective, the most scientific and reasonable method is the decision-theoretic approach because it is determined by certain economic factors, such as minimizing loss or maximizing return.

%The reliability or any lifetime characteristics of the product are key measured quality characteristics. It has a great impact on consumers' decisions. Therefore, in ASP, the measurement of the reliability or any lifetime characteristics of the product can be used to take the decision of acceptance or rejection a lot. This problem of ASP is called the reliability acceptance sampling problem (RASP).
Due to time and cost constraints, censored life tests are generally conducted in practice. Determination of RASP based on censored data is an important task in reliability studies. The most common type of censoring schemes used in practice are type I, type II and hybrid censoring schemes. Continuous monitoring of life test is required to obtain lifetime data under these schemes. However, in many situations, it is not possible to monitor the test continuously and the items on the test are inspected at some pre-specified inspection times. Here, the failure times are not observed but only the number of failed items is observed at each inspection time. The data are referred to as interval-censored data, and the test is known as the interval-censored test (see \citet{lu2009interval}). In ICS, all items are put on life test at $\tau_0=0$ and the items are inspected at pre-determined inspection times $\tau_1<\tau_2<\cdots<\tau_k$, where the inspection number, $k$, is also pre-determined. 

There have been several works for determining optimal BSPs for exponential distribution under different censoring schemes. \citet{yeh1990optimal} considered the optimal BSP under type-II censoring. A similar approach is considered determination of BSP under different censoring schemes by many authors. \citet{yeh1994bayesian} and  \citet{yeh1995bayesian} considered for Type-I and random censoring schemes, respectively. \citet{chen2004designing} and \citet{lin2008exact} provided methodology for hybrid censoring scheme. These BSPs for censored data have not considered the prior belief of the parameter in the decision function.

\citet{papazoglou1999bayesian} analyzed a Bayesian decision problem based on reliability of the product and pointed out that the decision of acceptance or rejection about the lot can be made based on the existing prior information. However, due to the uncertainty of the parameters, the consumer can not be sure about the decision regarding the lot. Additional information needs to be obtained through life-testing at a cost. The lifetime data can be combined with the existing prior information to update the information. The decision can be taken using the updated information on the product’s reliability. \citet{lin2002bayesian} considered a decision function called Bayes decision function under type-I censoring. This Bayes decision function was also studied for different censoring schemes. For example, \citet{liang2013optimal} and \citet{yang2017optimal} used Bayes decision function for the type-I hybrid and modified type-II hybrid censoring schemes, respectively. The Bayes decision function for the ICS scheme was provided by \citet{chen2015bayesian}.

 The BSPs referenced above primarily focused on systems that have a single failure mode. However, product may fail due to more than one cause called competing risks.  In the competing risk setup, the data consists of product lifetime and an indicator variable denoting the specific cause of the failure. If we ignore the information on causes of failure, then it may incorrectly infer a product's reliability. For details,  see \citet{crowder2001classical} and \citet{nelson2009accelerated}. To analyze the competing risk data, the latent failure rate model proposed by \citet{cox1959analysis} is the most popular model in reliability analysis. The cause of the failure may be statistically independent or dependent. It is generally assumed, however, that these competing risks are statistically independent even though the assumption of statistical dependence on the data is more realistic. Because in the assumption of dependence, the latent failure rate model may face some identifiability issues. To avoid such problems, it is assumed that latent failure times are statistically independent.
 
 The problem of designing RASP for competing risk data is an important task in reliability study. \citet{wu2014planning} studied RASP for competing risk data under progressive ICS schemes by the classical approach. Recently, \citet{prajapati2022optimal} studied BSP for competing risk data under type-II and hybrid censoring schemes. We consider designing optimal BSP for competing risk data under ICS.

The main contribution of the paper is to determine the optimal BSP by minimizing the Bayes risk in the case of exponential distribution under ICS. As discussed before, the consumer can takes the decision based on the product's reliability. Therefore, one of the purposes of this paper is to develop a BSP using the decision function based on the MLE of the RF. The second purpose of the paper is to approximate the BSP using the asymptotic properties of the MLE of the RF.  The third purpose of the paper is to develop a BSP using the Bayes decision function. 

The organization of the paper is as follows.  In Section \ref{model}, the model for exponential distribution in the presence of competing risks data under ICS is discussed and the Bayes risk for the general decision function is derived. In Section \ref{dec1}, the explicit form of the decision function is derived when the decision is taken based on the MLE of the RF and also an approximate Bayes risk is derived using the asymptotic properties of the MLE of the RF. In Section \ref{bayesa}, the Bayes decision function is derived. In Section \ref{num}, the numerical study is performed. In Section \ref{con}, the conclusion is made along with direction of future work. 
%The organization of the paper is as follows. At first, we discuss the model for exponential distribution when the item was inspected at $\tau_1,\tau_2,\cdots,\tau_k$ in section \ref{model}. In that section, we discuss the consumer's utility function, the prior distribution, the decision function, and the manufacturer's loss function. Also, the model is described. In section \ref{exp}, We determine the optimum RASPs when items are inspected at the equal length of the interval. The exact form of the decision function and the manufacturer's cost are obtained. In section \ref{exp1}, optimal BSP is discussed. The Bayes decision function is obtained by taking the same loss function and prior distribution. In Section \ref{num}, numerical results are presented for both cases.
\section{Model \& Assumptions}\label{model}
Suppose $n$ items are put on life test. $X_1,\ldots,X_n$ be the IID lifetimes of $n$ items. Let us consider that every item may fail due to any one of $J$ causes of failure. Let $X_{ij}$ be the latent lifetime of $i^{th}$ item which fails due to $j^{th}$ cause and therefore $X_i=\min\{X_{i1},X_{i2},\ldots,X_{iJ}\}$, $i=1,2, \ldots,n$. It is assumed that $X_{ij}$ follows an exponential distribution with PDF
\begin{align*}
    f_j(t\ | \ \nu_j)=\nu_j\exp(-\nu_j t),~\nu_j>0. 
\end{align*}
In the work,  $X_{i1},\ldots,X_{iJ}$ are assumed to be independent. Therefore $X_i$ follows exponential distribution with PDF \begin{align*}
      f(t\ | \ {\boldsymbol{\nu}})=\nu\exp(-\nu t), 
\end{align*}
CDF 
\begin{align*}
    F(t\ | \ {\boldsymbol{\nu}})=1-\exp(-\nu t)
\end{align*}
and RF
\begin{align*}
    \Bar{F}(t\ | \ {\boldsymbol{\nu}})=\exp(-\nu t),
\end{align*}
where $\nu=\nu_1+\cdots+\nu_J$ and $\boldsymbol{\nu}=(\nu_1,\nu_2,\ldots,\nu_J)$, for $i=1,\ldots,n$.
Now let $C_i$ denote the cause of the failure of the $i^{th}$ item. The sub-distribution of $X_i$ and $C_i$ is defined as
\begin{align*}
    G(j,t\ | \ {\boldsymbol{\nu}})=P(C=j,X_i\leq t)=\frac{\nu_j}{\nu}\left[1-\exp(-\nu t)\right]. 
\end{align*}
In the above scenario, $t>0$, $i=1,\ldots,n\text{ and }  j=1,\ldots,J$. Let us consider that the life test is conducted under ICS as described in Introduction. In ICS, we only observe for each cause, the number of failures at pre-specified inspection times $\tau_1,\ldots,\tau_k$, where $0<\tau_1<\cdots<\tau_k$. Suppose $D_{mj}$ denotes number of failures in the interval $(\tau_{m-1},\tau_m]$ due to the cause $j$, for $m=1,\ldots,k$ and $j=1,\ldots,J$. Let $\boldsymbol{D}=(D_{11},\ldots,D_{1J},D_{21},\ldots,D_{2J}$ $,\ldots,D_{m1},\ldots,D_{mJ})$. The observed data up to inspection time $\tau_k$ is $\boldsymbol{d}=(d_{11},\ldots,d_{1J},d_{21},\ldots,d_{2J},\ldots,d_{m1},\ldots,d_{mJ})$, where $d_{mj}$ is the observed value of $D_{mj}$. for $m=1,\ldots,k$ and $j=1,\ldots,J$. Therefore, $d_t=\sum_{m=1}^k\sum_{j=1}^Jd_{mj}$ is the total number of failures and $n-d_t$ is the survived items after the life test. Note that $\boldsymbol{D}$ have multinomial distribution with joint distribution
\allowdisplaybreaks\begin{align}\label{joint}
    P_{\boldsymbol{D}}(\boldsymbol{d}\ | \ \boldsymbol{\nu})&
    =\frac{n!}{\prod\limits_{m=1}^k\prod\limits_{j=1}^Jd_{mj}!(n-d_t)!}\prod_{m=1}^k\prod_{j=1}^J \left[G(j, \tau_{m}\ | \ {\boldsymbol{\nu}})-G(j,\tau_{m-1}\ | \ {\boldsymbol{\nu}})\right]^{d_{mj}}\left[\Bar{F}(\tau_k\ | \ {\boldsymbol{\nu}})\right]^{n-d_t}\nonumber\\
    &=\frac{n!}{\prod\limits_{m=1}^k\prod\limits_{j=1}^Jd_{mj}!(n-d_t)!}\prod_{m=1}^k\prod_{j=1}^J \left[\frac{\nu_j}{\nu}\left\{\exp(-\nu \tau_{m-1})-\exp(-\nu \tau_m)\right\}\right]^{d_{mj}}\nonumber\\
    &~~~~~\times\left[\exp(-\nu \tau_k)\right]^{n-d_t}.
\end{align} 
The vector of designing parameters for ICS is denoted by
$\boldsymbol{\zeta} = (n, \tau_1,\ldots,\tau_k, k)$ and based on the observed data $\boldsymbol{d}$, the decision function $a(\boldsymbol{d}\ | \ \boldsymbol{\zeta})$ is defined as
\begin{align}
    a(\boldsymbol{d}\ | \ \boldsymbol{\zeta})=\begin{dcases}
        1  & \text{if accepting the lot when data $\boldsymbol{d}$ is observed} ~~\\
        0 & \text{if rejecting the lot when data $\boldsymbol{d}$ is observed},
    \end{dcases}
\end{align}
However, whether the lot is accepted or rejected, there is a loss associated with the decision.  If a bad lot is accepted, the manufacturer can bear an acceptance cost $h(\boldsymbol{\nu})$. We take the loss function as $h(\boldsymbol{\nu})=C_0+\sum_{j=1}^JC_j\nu_j+\mathop{\sum_{i=1}^J\sum_{j=1}^J}\limits_{\substack{i\leq j}}C_{ij}\nu_i\nu_j$, where the coefficient $C_i$ and $C_{jl}$, for $i=0,1,\ldots,J$, $j=1,\ldots,l$ and $l=1,\ldots,J$ are taken in such a way that $h(\boldsymbol{\nu})\geq 0$  $\forall$  $\boldsymbol{\nu}>\boldsymbol{0}$. This loss function is known as a quadratic loss function. If the lot is rejected, the lot is returned or discarded. So the cost due to rejection is fixed. Let $C_r$ be the rejection cost. The loss function is considered as
\begin{align}\label{acc}
\mathcal{L}\left(a(\boldsymbol{d}\ | \ \boldsymbol{\zeta})\ | \  \boldsymbol{\nu}\right)=
\begin{dcases*}
    h(\boldsymbol{\nu})+nC_s+\tau C_{\tau}+MC_I-(n-d_t)r_s& if $a(\boldsymbol{d}\ | \ \boldsymbol{\zeta})\ | \  \boldsymbol{\nu})=1$\\
    C_r+nC_s+\tau C_\tau+MC_I-(n-d_t)r_s& if $a(\boldsymbol{d}\ | \ \boldsymbol{\zeta})\ | \  \boldsymbol{\nu})=0$,
\end{dcases*} 
\end{align}
where $C_s$ is the cost of each item in the sample,  $C_{\tau}$ is the cost per unit time, $r_s~(<C_s)$ is the salvage value of each item that survived after life test and $C_I$ is the cost of each inspection, $\tau$ is the duration of life test, $M$ is the inspection number and $d_t$ is the total number of failures in the life testing experiment.\\
$\mathcal{L}(a(\boldsymbol{d}\ | \ \boldsymbol{\zeta})\ | \ \boldsymbol{\nu})$ in  (\ref{acc}) can be expressed as 
\allowdisplaybreaks\begin{align*}
   \mathcal{L}\left(a(\boldsymbol{d}\ | \ \boldsymbol{\zeta})\ | \  \boldsymbol{\nu}\right)=&\left[h(\boldsymbol{\nu})+nC_s+\tau C_{\tau}+MC_I-(n-d_t)r_s\right]a(\boldsymbol{d}\ | \ \boldsymbol{\zeta})\\
   &+\left[C_r+nC_s+\tau C_{\tau}+MC_I-(n-d_t)r_s\right]\left[1-a(\boldsymbol{d}\ | \ \boldsymbol{\zeta})\right]\\
   =&C_r+nC_s+\tau C_{\tau}+MC_I-(n-d_t)r_s+a(\boldsymbol{d}\ | \ \boldsymbol{\zeta})\left[h(\boldsymbol{\nu})-C_r\right].
\end{align*}
Now, the Bayes risk is denoted by $R(\boldsymbol{\zeta},a)$ and defined by 
\begin{align}\label{R}
R(\boldsymbol{\zeta},a)&=E_{\boldsymbol{\nu}}E_{\boldsymbol{d}\ | \ \boldsymbol{\nu}}\left[\mathcal{L}\left(a(\boldsymbol{d}\ | \ \boldsymbol{\zeta})\ | \  \boldsymbol{\nu}\right)\right]\nonumber\\
&=C_r+n(C_s-r_s)+C_{\tau} E_{\boldsymbol{\nu}}[\tau\ | \ \boldsymbol{\nu}]+C_IE_{\boldsymbol{\nu}}[M\ | \ \boldsymbol{\nu}]+r_sE_{\boldsymbol{\nu}}[D_t\ | \ \boldsymbol{\nu}]+R_1(\boldsymbol{\zeta},a)\nonumber\\
&=C_r+n(C_s-r_s)+C_{\tau} E[\tau]+C_IE[M]+r_sE[D_t]+R_1(\boldsymbol{\zeta},a),
\end{align}
where
\begin{align}\label{R1}
R_1(\boldsymbol{\zeta},a)
&=E_{\boldsymbol{\nu}}E_{\boldsymbol{d}\ | \ \boldsymbol{\nu}} [a(\boldsymbol{d}\ | \ \boldsymbol{\zeta})[h(\boldsymbol{\nu})-C_r]\nonumber\\
&=\sum_{\boldsymbol{d}\in\mathcal{X}}\int_{\boldsymbol{\nu}} [h(\boldsymbol{\nu})-C_r]P_{\boldsymbol{D}}(\boldsymbol{d}\ | \ \boldsymbol{\nu})p(\boldsymbol{\nu})~ d\boldsymbol{\nu}
\end{align}
and $\mathcal{X}$ is the set of all possible $\boldsymbol{d}$ for which the lot is accepted after life testing. Now, we have $E[D_t]=E_{\boldsymbol{\nu}}[E[D_t\ | \ \boldsymbol{\nu}]]=E_{\boldsymbol{\nu}}[n(1-\exp(-\nu \tau_k))]$, $E[\tau]=E_{\boldsymbol{\nu}}[\tau \ | \ \boldsymbol{\nu}]=\sum_{i=1}^k \tau_i P_i$,
where $\mathcal{P}_i= P_i-P_{i-1}$ is the probability that the life-test terminates at $\tau_i$ 
with $P_i=\left[1-\exp(-\nu \tau_i)\right]^n,$
for $i=1,\ldots,k-1$ and $P_k=1$
and $E[M\ | \ \boldsymbol{\nu}]=\sum_{i=1}^k i ~\mathcal{P}_i$.

 \noindent Now, the optimal BSP $(\boldsymbol{\zeta}^*,a^*)=(n^*,\tau_1^*,\tau_2^*,\ldots,\tau_k^*,k^*,a^*)$ is obtained by maximizing $R(\boldsymbol{\zeta},a)$, i.e,
$(\boldsymbol{\zeta}^*,a^*)=\arg\max\limits_{(\boldsymbol{\zeta},a)}R(\boldsymbol{\zeta},a)$.
Note that for equal length of interval $h$, the optimal BSP is $(\boldsymbol{\zeta}^*,a^*)=(n^*,h^*,k^*,a^*)$.
\section{BSP based on reliability criteria decision function}\label{dec1}
Here we obtain BSP under the assumption that the manufacturer knows the consumer's acceptance criteria. The consumer takes his/her decision based on the reliability criteria. A lot is of acceptable quality when the product's reliability at a specified time point $\tau_0$, $\Bar{F}(\tau_0\ | \ {\boldsymbol{\nu}})$ is greater than $R_0$. 
Let $\widehat{\boldsymbol{\nu}}$ is the MLE of ${\boldsymbol{\nu}}$. Therefore, $\Bar{F}(\tau_0\ | \ \widehat{{\boldsymbol{\nu}}})$ is the MLE of $R(\tau_0\ | \ {\boldsymbol{\nu}})$. Now, the lot is accepted if   ${{R}}(\tau_0\ | \ \widehat{\boldsymbol{\nu}})>R_0$ and rejected if $\Bar{F}(\tau_0\ | \ \widehat{\boldsymbol{\nu}})\leq R_0$.
 Now, based on the competing risks data under ICS the likelihood function is given by
\begin{align}\label{like}
     L(\boldsymbol{\nu} \ | \ \boldsymbol{d})&\propto\prod_{m=1}^k\prod_{j=1}^J[G(j,\tau_m\ | \ {\boldsymbol{\nu}})-G(j,\tau_{m-1}\ | \ {\boldsymbol{\nu}})]^{d_{mj}}[\Bar{F}(\tau_{k}\ | \ {\boldsymbol{\nu}})]^{n-d_t}\nonumber\\
   &\propto\prod_{m=1}^k\prod_{j=1}^J \left[\frac{\nu_j}{\nu}\left\{\exp(-\nu \tau_{m-1})-\exp(-\nu \tau_m)\right\}\right]^{d_{mj}}\left[\exp(-\nu \tau_k)\right]^{n-d_t}\nonumber\\
   &\propto \prod_{m=1}^k\prod_{j=1}^J \left[\exp(-\nu \tau_{m-1})\right]^{d_{mj}}\left[\frac{\nu_j}{\nu}\left\{1-\exp(-\nu (\tau_m-\tau_{m-1}))\right\}\right]^{d_{mj}}\left[\exp(-\nu \tau_k)\right]^{n-d_t}.
\end{align}
The log-likelihood function is given by
\begin{align*}
    l(\boldsymbol{\nu}\ | \ \boldsymbol{d})=&\sum_{m=1}^k\sum_{j=1}^J\left[-d_{mj}\left\{\nu \tau_{m-1}+\log(\nu_j)-\log(\nu)+\log(1-\exp(-\nu(\tau_m-\tau_{m-1})))\right\}\right]\\
    &-(n-d_t)\nu \tau_k.
\end{align*}
 The likelihood equations are written as
\begin{align}\label{li}
   & \frac{\partial  l(\boldsymbol{\nu}\ | \ \boldsymbol{d}) }{\partial\nu_j}=0\nonumber\\
    \implies&\sum_{m=1}^k\frac{d_{mj}}{\nu_j}-\sum_{m=1}^kd_{m+}\left[\tau_{m-1}+\frac{1}{\nu}-\frac{(\tau_m-\tau_{m-1})\exp(-\nu(\tau_m-\tau_{m-1}))}{1-\exp(-\nu(\tau_m-\tau_{m-1}))}\right]-(n-d_t)\tau_k=0,\nonumber\\
    \implies&\frac{d_{+j}}{\nu_j}=\frac{d_t}{\nu}+g(\nu),
\end{align}
where $g(\nu)=\sum_{m=1}^kd_{m+}\left[\tau_{m-1}-\frac{(\tau_m-\tau_{m-1})\exp(-\nu(\tau_m-\tau_{m-1}))}{1-\exp(-\nu(\tau_m-\tau_{m-1}))}\right]+(n-d_t)\tau_k$, $d_{m+}=\sum_{j=1}^J d_{mj}$, $d_{+j}=\sum_{m=1}^kd_{mj}$ and $d_t=\sum_{m=1}^kd_{m+}=\sum_{j=1}^Jd_{+j}$, for $j=1,\cdots,J$. Solving  (\ref{li}), for $j=1,\ldots,J$, we get the MLEs $\widehat{\boldsymbol{\nu}}=(\widehat{\nu}_1,\ldots,\widehat{\nu}_J)$. Now, from the likelihood equations, we get
\begin{align}\label{re}
   & \frac{\widehat{\nu}_1}{d_{+1}}=\cdots=\frac{\widehat{\nu}_J}{d_{+J}}=\frac{\widehat{\nu}}{d_t}\nonumber\\
    \implies &\widehat{\nu}_{j}=\frac{d_{+j}\widehat{\nu}}{d_t},~~~~~j=1,\ldots,J,
\end{align}
where $\widehat{\nu}=\widehat{\nu}_1+\cdots+\widehat{\nu}_J$. From  (\ref{li}) using  (\ref{re}), we get $g(\widehat{\nu})=0$. It is easy to show that $\lim_{\nu\rightarrow0}g(\nu)=-\infty$ and $\lim_{\nu\rightarrow\infty}=\sum_{m=1}^kd_{m+}\tau_{m-1}+(n-d_t)\tau_k>0$. Also, we can see that $g(\nu)$ is a strictly increasing function in $\nu$. Hence, we get a unique solution of $\widehat{\nu}$. 
The decision function can be expressed as
\allowdisplaybreaks\begin{align}
    a(\boldsymbol{d}\ | \ \boldsymbol{\zeta})=\begin{dcases}
        1  & \exp(-\widehat{\nu}\tau_0)>R_0~~\\
        0 & \text{Otherwise},
    \end{dcases}
\end{align}
where $\widehat{\nu}$ is the root of the equation $g(\nu)=0$. Note that we only need to estimate $\nu$ for our decision function. It is seen that $\widehat{\nu}$ exists when $d_t>0$. When $d_t=0$, $\widehat{\nu}$ is taken as $1/n\tau_k$. If the inspection times are of equal length. i.e, $\tau_m=mh$, $m=1,\cdots,k$, where $h$ is the length of each interval. In that case, the explicit form of the MLE of ${\nu}$ is given by
\allowdisplaybreaks\begin{align}
    \widehat{\nu}=\begin{dcases*}
        -\frac{1}{h}\ln\left(1-\frac{d_t}{\sum_{m=1}^k md_{m+}+(n-d_t)k}\right)&$d_t>0$\\
        \frac{1}{nkh}&$d_t=0$.
    \end{dcases*} 
\end{align}
As a result, the decision function is expressed as if $d_t>0$,
\begin{align*}
    a(\boldsymbol{d}\ | \ \boldsymbol{\zeta})=\begin{dcases*}1 & if
       $\left(\frac{d_t}{\sum_{m=1}^k md_{m+}+(n-d_t)k}\right)^{\tau_0/h}> R_0$ \\
       0 & otherwise
    \end{dcases*}
\end{align*}
and if $d_t=0$,
\begin{align*}
    a(\boldsymbol{0}\ | \ \boldsymbol{\zeta})=\begin{dcases*}1 & if
       $\exp\left(-\frac{\tau_0}{nkh}\right)> R_0$ \\
       0 & otherwise.
    \end{dcases*}
\end{align*}
 In this case, the form of the decision function is known. Therefore, to find the optimal decision we need to determine the optimal value of $R_0$. Let $R_0^*$ be the optimal value of $R_0$. In that case $a^*\equiv R_0^*$. 
  Now, it is assumed that the failure rate of the items $\nu=\nu_1+\cdots+\nu_J$ follows gamma distribution with PDF 
\begin{align*}
    p(\nu)=\frac{\eta^{\alpha}}{\Gamma(\alpha)}\nu^{\alpha-1}\exp(-\eta\nu), ~~~\nu>0, \alpha>0, \eta>0
\end{align*}
and it is denoted by $G(\alpha,\eta)$. It is assumed that for given $\nu$, $(\nu_1/\nu,\ldots,\nu_J/\nu)$ follows Dirichlet distribution with JPDF 
\begin{align*}
p(\nu_1/\nu,\ldots,\nu_J/\nu)=\frac{\Gamma(\alpha_0)}{\prod_{j=1}^J\Gamma(\alpha_j)}\prod_{j=1}^J\left(\frac{\nu_1}{\nu}\right)^{\alpha_j-1}, ~~0<\frac{\nu_j}{\nu}<1, \alpha_j>0, \text{ for } j=1,\ldots,J,
\end{align*}
where $\nu_j/\nu$ is the ratio of the failure rate that fails due to $j^{th}$ cause and the failure rate of the item and $\alpha_0=\sum_{j=1}^J\alpha_j$. Therefore, the prior JPDF of $\boldsymbol{\nu}=(\nu_1,\ldots,\nu_J)$ is
\begin{align}
    p(\boldsymbol{\nu})= \frac{\eta^{\alpha}}{\Gamma(\alpha)}\nu^{\alpha-\alpha_0}\exp(-\eta\nu)\prod_{j=1}^J
    \frac{\Gamma(\alpha_0)}{\Gamma(\alpha_j)}\nu_j^{\alpha_j-1}.
\end{align}
This is known as the Gamma-Dirichlet distribution and is denoted by $GD(\alpha,\eta,\alpha_1,\ldots,\alpha_J)$. We provide the following results needed for computation of Bayes risk.
\begin{result}\label{r1}
     \begin{align*}
\int_{\nu_1=0}^\infty \cdots\int_{\nu_J=0}^\infty\nu^{\alpha-\alpha_0}\exp(-\eta \nu)\prod_{j=1}^J \nu_j^{\alpha_j-1}~d\nu_1\cdots d\nu_J=\frac{\Gamma(\alpha)}{\eta^\alpha\Gamma(\alpha_0)}\prod_{j=1}^J\Gamma(\alpha_j).
 \end{align*}
\end{result}
        \begin{result}\label{r2}
 If   a sample of size $n$ is put on the test  and inspected at an equal length of the interval $h$ up to the maximum number of inspections $k$, then the Bayes risk can be expressed as
 \begin{align*}R(\boldsymbol{\zeta},R_0)=C_r+nC_s+r_sn\left(\frac{\eta}{\eta+kh}\right)^{\alpha}+(C_I+C_{\tau}h)\left[k-\sum_{i=1}^{k-1}\sum_{j=1}^n\binom{n}{j}(-1)^j\left(\frac{\eta}{\eta+ijh}\right)^\alpha\right]+R_1(\boldsymbol{\zeta},R_0),
 \end{align*}
 where
\begin{align*}
    R_1(\boldsymbol{\zeta},R_0)=& \sum_{\boldsymbol{d}\in \mathcal{X}} \frac{n!}{d_{11}!d_{12}!\cdots d_{m1}!d_{k2}!(n-d_t)!} \sum_{i=0}^{d_t}(-1)^i\binom{d_t}{i}\left[(C_0-C_r)\frac{\eta^\alpha}{(\eta+ih+sh)^\alpha}\right.\\
    &\left.+{\Gamma(\alpha_0)}\left\{\frac{\alpha\eta^\alpha}{(\eta+ih+sh)^{\alpha+1}}\sum_{p=1}^J C_p\prod_{j=1}^J\frac{\Gamma(\alpha_j+d_{+j}+\delta_{pj})}{\Gamma(\alpha_j)\Gamma(\alpha_0+d_t+1)}\right.\right.\\
     +&\left.\left.\frac{\alpha(\alpha+1)\eta^\alpha}{(\eta+ih+sh)^{\alpha+2}}\underset{p\leq q}{\sum_{p=1}^J \sum_{q=1}^J}C_{pq}\prod_{j=1}^J\frac{\Gamma(\alpha_j+d_{+j}+\delta_{pj}+\delta_{qj})}{\Gamma(\alpha_j)\Gamma(\alpha_0+d_t+2)}\right\}\right]
    \end{align*}
    and $\delta_{ij}=\begin{dcases*}
            1 & \text{if} $i=j$\\
            0 & \text{otherwise}.
        \end{dcases*}$
\end{result} 
 \begin{proof}
 The expressions of $E[D_t]$, $E[M]$ and $E[\tau]$ are given in \citet{chen2015bayesian}. We only need to prove the term $R_1(\boldsymbol{\zeta},R_0)$.
 Using  (\ref{joint}),  when the inspection intervals are equal, the joint distribution of $\boldsymbol{D}$ can be expressed 
\allowdisplaybreaks\begin{align}
    P_{\boldsymbol{D}}( \boldsymbol{d}\ | \ \boldsymbol{\nu})
&=\frac{n!}{\prod\limits_{m=1}^k\prod\limits_{j=1}^Jd_{mj}!(n-d_t)!}\prod_{m=1}^k\prod_{j=1}^J\left\{\exp\left[-\nu d_{mj}(m-1)h\right]\left[1-\exp(-\nu h)\right]^{d_{mj}}\left(\frac{\nu_j}{\nu}\right)^{d_{mj}}\right\}\nonumber\\
&~~~\times\exp\left[-\nu(n-d_t)kh\right]\nonumber\\
&=\frac{n!}{\prod\limits_{m=1}^k\prod\limits_{j=1}^Jd_{mj}!(n-d_t)!}\sum_{i=0}^{d_t}(-1)^i \binom{d_t}{i}\prod_{j=1}^J\left(\frac{\nu_j}{\nu}\right)^{d_{+j}}\exp\left[-{(i+s)\nu_j h}\right],
\end{align}
where $s=nk-\sum_{m=1}^k\sum_{j=1}^J(k-(m-1))d_{mj}$. \\
We know that    $R_1(\boldsymbol{\zeta},a)=\sum_{\boldsymbol{d}\in \mathcal{X}}\int_{\boldsymbol{\nu}}\left[h(\boldsymbol{\nu})-C_r\right]~P_{\boldsymbol{D}}(\boldsymbol{d}\ | \ \boldsymbol{\nu})~p(\boldsymbol{\nu}) ~d\boldsymbol{\nu}$.
Now using Result \ref{r1}, we get the required expression of $R_1(\boldsymbol{\zeta},a)$. 
 \end{proof}\\
 It is noted that when length of inspection intervals are equal, the vector of decision variable can be expressed as $\boldsymbol{\zeta}=(n,h,k)$. 
For different inspection intervals, the term $R_1(\boldsymbol{\zeta},a)$ cannot be obtained analytically, a simple Monte Carlo integration can be used to evaluate the integration. For this,  a large number $N_1$ of observations are generated from the prior distributions $p(\boldsymbol{\nu})$.  If $p(\boldsymbol{\nu})$ follows $GD(\alpha,\eta,\alpha_1,\ldots,\alpha_J)$, then the data can be generated using Algorithm 1.\\
\textbf{Algorithm 1: }
\begin{enumerate}[I.]
    \item Generate $\nu^{(i)}$ from gamma distribution with parameters $(\alpha,\eta)$.
    \item For given $\nu^{(i)}$, generate
 $\frac{\boldsymbol{\nu}^{(i)}}{\nu^{(i)}}$ from Dirichlet distribution with parameters $(\alpha_1,\ldots,\alpha_J)$ using the built-in R packages "LaplacesDemon".
    \item Multiplying $\nu^{(i)}$ and $\frac{\boldsymbol{\nu}^{(i)}}{\nu^{(i)}}$, we get the required generated sample $\boldsymbol{\nu}^{(i)}$. 
\end{enumerate}
Let $\boldsymbol{\nu}^{(1)}\ldots,\boldsymbol{\nu}^{(N_1)}$ be the generated observations. Then 
     \begin{align*}
R_1(\boldsymbol{\zeta},a)\approx\frac{1}{N_1}\sum_{i=1}^{N_1}\left[C_r-h\left(\boldsymbol{\nu}^{(i)}\right)\right]\sum_{\boldsymbol{d}\in\mathcal{X}}P_{\boldsymbol{D}}\left(\boldsymbol{d}\ | \ \boldsymbol{\nu}^{(i)}\right).
    \end{align*}
    \subsection{Approximate form of Bayes risk using asymptotic properties of the MLE}\label{dec2}
For large samples and different prior distributions, computing Bayes risk will be intensive and complex. Due to complexity, we approximate the Bayes risk in a simplified form, where we can calculate Bayes risk for any choice of prior. Under certain regularity conditions, for large $n$, $\widehat{\boldsymbol{\nu}}$ follows normal distribution with mean $\boldsymbol{\nu}$ and variance-covariance matrix $I^{-1}(\boldsymbol{\nu})$, where $I^{-1}(\boldsymbol{\nu})$ is inverse of the FIM of $\boldsymbol{\nu}$. The regularity conditions are stated below:\\
\textbf{Regularity Conditions (\citet{das2024optimalplanningprogressivetypei}): } 
\allowdisplaybreaks\begin{enumerate}[I.]
    \item $X_1,X_2,\ldots,X_n$ are IID unobserved lifetimes with common PDF $f(\cdot \ | \ {\boldsymbol{\nu}})$ and  CDF $F(\cdot \ | \ {\boldsymbol{\nu}})$.
    \item The number of causes of failure is finite.
\item  $(X_1,C_1),(X_2,C_2),\ldots,(X_n,C_n)$ are IID with common sub-distribution function $G(\cdot,\cdot \ | \ {\boldsymbol{\nu}})$.
\item The supports of $f(\cdot \ | \ {\boldsymbol{\nu}})$ and $G(\cdot \ | \ {\boldsymbol{\nu}})$ is independent of $\boldsymbol{\nu}$.
    \item In the parameter space $\boldsymbol{\nu}$, the true parameter $\boldsymbol{\nu}^0$ is an interior point of the open set $\boldsymbol{\nu}_0$.
    \item For almost all $t$ and  $\forall$  $j=1,\ldots,J$, $F(t\ | \ {\boldsymbol{\nu}})$ and $G(j,t)\ | \ {\boldsymbol{\nu}})$ admit all third order derivatives $\frac{\partial^3 F(t\ | \ {\boldsymbol{\nu}})}{\partial\nu_u\partial\nu_v\partial\nu_w}$ and $\frac{\partial^3 G(j,t\ | \ {\boldsymbol{\nu}})}{\partial\nu_u\partial\nu_v\partial\nu_w}$, respectively  $\forall$  $\boldsymbol{\nu}\in \boldsymbol{\nu}_0$ and $w,u,v=1,\ldots,J$. Also, upto third-order derivatives of $G(j,t\ | \ {\boldsymbol{\nu}})$ and $F(t\ | \ {\boldsymbol{\nu}})$ wrt $\boldsymbol{\nu}$ are bounded  $\forall$  $\boldsymbol{\nu}\in\boldsymbol{\nu}_0$.
    \item In such a way, the $\tau_m$'s are chosen that
    \begin{enumerate}[(a)]
        \item $0<q_{mj}<1$ and $0<q_m<1$, where
      \begin{align}\label{qij}
    q_{mj}=\frac{P\left(X_i\le \tau_m,I=j\right)-P\left(X_i\le \tau_{m-1},I=j\right)}{P\left(X_i\geq \tau_{m-1}\right)}&=\frac{{G}\left(j,\tau_{m}\ | \ {\boldsymbol{\nu}}\right)-{G}\left(j,\tau_{m-1}\ | \ {\boldsymbol{\nu}}\right)}{\Bar{F}(\tau_{m-1}\ | \ {\boldsymbol{\nu}})}
\end{align} is the probability of an item failed by the time $\tau_m$ due to cause $j$ given that the item is at risk at time $\tau_{m-1}$ and \begin{align}\label{qi}
q_m=\frac{P\left(X_i \le \tau_{m}\right)-P\left(X_i \le \tau_{m-1}\right)}{P\left(X_i\geq \tau_{m-1}\right)}=\frac{{F}(\tau_{m}\ | \ {\boldsymbol{\nu}})-{F}(\tau_{m-1}\ | \ {\boldsymbol{\nu}})}{\Bar{F}(\tau_{m-1}\ | \ {\boldsymbol{\nu}})}.    
\end{align}
is the probability of an item failed by the time $\tau_m$ given that the item is at risk at time $\tau_{m-1}$, for $i=1,\ldots,n$, $j=1,\ldots,J$ and $m=1,\ldots,k$.
        \item $\nabla_{\boldsymbol{\nu}}\boldsymbol{q}$ is matrix of rank $l$, where $\nabla_{\boldsymbol{\nu}}\boldsymbol{q}=\left(\frac{\partial q_{i}}{\partial\nu_u}\right)_{s\times m}$ for $u=1,\ldots,J$ and $i=1,\ldots,m$.
    \end{enumerate}
    \end{enumerate}
    The FIM is given by $I(\boldsymbol{\nu})=(I_{u,v}(\boldsymbol{\nu}))$, where
\begin{align}
    I_{u,v}(\boldsymbol{\nu})=E\left[-\frac{\partial ^2 l(\boldsymbol{\nu})}{\partial\nu_u\partial\nu_v}\right]
\end{align}
 For interval-censored data, the FIM for $\boldsymbol{\nu}$ can be expressed as
    \begin{align}\label{fisher}
I(\boldsymbol{\nu})=\sum_{m=1}^k\left[\sum_{j=1}^J\frac{n[R(\tau_{m-1}\ | \ {\boldsymbol{\nu}})]} {q_{mj}}\nabla_{\boldsymbol{\nu}}(q_{mj})\nabla_{\boldsymbol{\nu}}(q_{mj})^T+\frac{n[R(\tau_{m-1}\ | \ {\boldsymbol{\nu}})]}{(1-q_{m})}\nabla_{\boldsymbol{\nu}}(q_{m})\nabla_{\boldsymbol{\nu}}(q_{m})^T\right].
    \end{align}
 The derivation of the form of the FIM and proof of the normal approximation of the MLEs are given in \citet{das2024optimalplanningprogressivetypei}.

Now we find the approximate expression of $P(\exp(-\widehat{\nu}\tau_0)>R_0)$. By using the delta method, we have    $c(\widehat{\boldsymbol{\nu}})\sim \mathcal{N}(c(\boldsymbol{\nu}),S(\boldsymbol{\nu})^2) $, where $c(\boldsymbol{\nu})=\exp(-\nu \tau_0)$, $S(\boldsymbol{\nu})^2=\nabla_{\boldsymbol{\nu}}c(\boldsymbol{\nu})^TI^{-1}(\boldsymbol{\nu})\nabla_{\boldsymbol{\nu}}c(\boldsymbol{\nu})$ and $\nabla_{\boldsymbol{\nu}}c(\boldsymbol{\nu})=\left(\frac{\partial c(\boldsymbol{\nu})}{\partial \nu_1},\ldots,\frac{\partial c(\boldsymbol{\nu})}{\partial \nu_J}\right)$.
    Therefore,
    \begin{align*}
     & P\left(\exp(-\widehat{\nu}\tau_0)>R_0\right)=  P\left(\frac{c(\widehat{\boldsymbol{\nu}})-c(\boldsymbol{\nu})}{S(\boldsymbol{\nu})}>\frac{R_0-c(\boldsymbol{\nu})}{S(\boldsymbol{\nu})}\right)\approx1-\Phi\left(\frac{R_0-c(\boldsymbol{\nu})}{S(\boldsymbol{\nu})}\right)=\Phi\left(\frac{c(\boldsymbol{\nu})-R_0}{S(\boldsymbol{\nu})}\right),
    \end{align*}
    where $\Phi(\cdot)$ is the CDF of standard normal distribution.\\
   For computation of $S(\boldsymbol{\nu})$, we need to compute $\nabla_{\boldsymbol{\nu}}(q_{mj})$ and $\nabla_{\boldsymbol{\nu}}(q_{m})$ to obtain $I\left(\boldsymbol{\nu}\right)$. For exponential  distribution, $q_{mj}$ given in (\ref{qij}) can be written as $q_{mj}=\frac{\nu_j}{\nu}\left[1-\exp\left\{-\nu\left(\tau_m-\tau_{m-1}\right)\right\}\right]$ and $q_m$ given in (\ref{qi}) can be written as $q_i=1-\exp\left[-\nu\left(\tau_m-\tau_{m-1}\right)\right]$, for $j=1,\ldots,J$ and $m=1,\ldots,k$. Therefore, 
     \begin{align*}
        \frac{\partial q_{mj}}{\partial\nu_l}=\begin{dcases*}
         \frac{\nu_j (\tau_m-\tau_{m-1})(1-q_m)}{\nu}-\frac{q_m\nu_j}{\nu^2},&$l\neq j$\\
          \frac{\nu_j (\tau_m-\tau_{m-1})(1-q_m)}{\nu}+\frac{q_m(\nu-\nu_j)}{\nu^2},& $l=j$,
        \end{dcases*}
    \end{align*}
    \begin{align*}
        \frac{\partial q_m}{\partial\nu_j}=(\tau_m-\tau_{m-1})\exp[-\nu(\tau_m-\tau_{m-1})]
    \end{align*}
 and 
     \begin{align*}
        \frac{\partial c(\boldsymbol{\nu})}{\partial\nu_j}=\tau_0\exp[-\nu \tau_0],
    \end{align*}
    for $m=1,\ldots,k$ and $l,j=1,\ldots,J$.
  \subsection{Finding optimal BSP} 
  Now, we provide upper bounds of $n$, $k$ and last inspection time point $\tau_k$. This shows that we get finite value of $\boldsymbol{\zeta}^*$. Since all cost-coefficients are non-negative and $C_s\geq r_s\geq 0$, this implies that $R(\boldsymbol{\zeta}^*,a)$ is non-negative. Therefore, the following equalities hold:
 \begin{align}
     R(\boldsymbol{\zeta}^*,a)\geq& ~n^*(C_s-r_s)+C_{\tau}E[\tau]+C_IE[I]\nonumber \\
     \geq &~n^*(C_s-r_s)+C_{\tau}\tau_k+C_IE[I]\label{i1}\\
     \geq& ~n^*(C_s-r_s)+C_Ik\label{i2}\\
     \geq&~ n^*(C_s-r_s) \label{i3}
 \end{align}
When the decision of acceptance or rejection about the lot is taken without life testing, that scenario is called no sampling.  For no sampling case, when the lot is accepted, $R(\boldsymbol{0},1)=E_{\nu}[h(\nu)]$ and when the lot is rejected, $R(\boldsymbol{0},0)=C_r$. Therefore,
 $R(\boldsymbol{0},a)=\min\{E_{\nu}[h(\nu)], C_r\}$. Now,
 \begin{align}\label{i4}
     R(\boldsymbol{\zeta}^*,a)\leq \min\{E_{\nu}[h(\nu)], C_r\}.
 \end{align}
 From (\ref{i3}) and (\ref{i4}), we get $$n^*\leq \frac{\min\{E_{\nu}[h(\nu)], C_r\}}{(C_s-r_s}.$$ Therefore, $n_0$ can be taken as $$n_0=\left\lfloor \frac{\min\{E_{\nu}[h(\nu)], C_r\}}{C_s-r_s}\right\rfloor,$$ where $\lfloor x\rfloor$ is the greatest integer less than or equal to x.\\
 From (\ref{i2}) and (\ref{i4}), we get $$k^*(n)= \frac{\min\{E_{\nu}[h(\nu)], C_r\}-(C_s-r_s)n}{C_I}.$$ Therefore, $k_0(n)$ can be taken as $$k_0(n)=\left\lfloor\frac{\min\{E_{\nu}[h(\nu)], C_r\}-(C_s-r_s)n}{C_I}\right\rfloor.$$
 From (\ref{i1}) and (\ref{i4}), we get $\tau_k^*(n,k(n))=\lfloor(\min\{E_{\nu}[h(\nu)], C_r\}-(C_s-r_s)n-C_Ik)/C_{\tau}$. Therefore, $\tau_k(n,k(n))$ can be taken as $$\tau_k(n,k(n))=\left\lfloor\frac{\min\{E_{\nu}[h(\nu)], C_r\}-(C_s-r_s)n-C_Ik}{C_{\tau}}\right\rfloor$$
 We consider the following algorithm to determine the optimal BSP. \\
\noindent \textbf{Algorithm 2:} 
\begin{enumerate}
    \item For each pair of values of $(n,k(n))$,  (\ref{R}) is minimized wrt $(\boldsymbol{\tau},R_0)$, where $\boldsymbol{\tau}=(\tau_1,\ldots,\tau_k)$. Since $\tau_1<\tau_2< \ldots< \tau_k<\tau_0(n,k(n))$, we consider a vector $\boldsymbol{h}=(h_1,\ldots,h_k)$ such that $\tau_1=h_1$, $\tau_{i}=\tau_{i-1}+h_i$, for $i=2,\ldots,k$ and $\sum_{i=1}^kh_i<\tau_0(n,k(n))$. For equal length of intervals, we take $h_i=h$, for $i=1,\ldots,k$ and $h<\tau_0(n,k(n))/k(n)$. Now,  minimize  (\ref{R}) wrt $\boldsymbol{h}$, where $h_i>0$, for $i=1,\ldots,k$, $R_0>0$ and $\sum_{i=1}^kh_i<\tau_0(n,k(n))$. The optimal $\boldsymbol{\tau}$ is given by $\boldsymbol{\tau}^*(n,k(n))=(\tau_1^*(n,k(n)),\tau_2^*(n,k(n)),\ldots,\tau^*_k(n,k(n))$, where for $i=1,2,\ldots,k$, $\tau^*_{i}(n,k)=\tau^*_{i-1}(n,k)+h^*_{i}(n,k)$. Let $h_i^*(n,k(n))$ denotes optimal value of $h_i(n,k(n))$, $i=1,\ldots,k$ and $\tau^*_{0}(n,k)=0$. Then we have
    \begin{flalign*}
       &R(n,\boldsymbol{\tau}^*(n,k(n)),k,R_0^*(n,k(n))) =\min_{(\boldsymbol{\tau},R_0),\tau_1< \cdots,\tau_k<\tau_0(n,k(n))} R(n,\boldsymbol{\tau},k,R_0). &&
    \end{flalign*}
    \item By minimizing  (\ref{R}) wrt $k$, we find an integer $k^*(n)$, $0\leq k^*(n)\leq k_0(n)$, for each value of $n$. i.e.,
    \begin{flalign*}
       & R(n,\boldsymbol{\tau}(n,k^*(n)),k^*(n),R_0^*(n,k^*(n)))=\min_{0\leq k\leq k(n)} R(n,\boldsymbol{\tau}^*(n,k(n)),k(n),R_0^*(n,k(n))).&&
    \end{flalign*}
    \item Finally, by minimizing  (\ref{R}), we find $n^*$, $0\leq n^*\leq n_0$. i.e.,
    \begin{flalign*}
       & R(n^*,\boldsymbol{\tau}^*(n^*,k^*(n^*)),k^*(n^*),R_0^*(n^*,k^*(n^*)))=\min_{0<n\leq n_0} R(n,\boldsymbol{\tau}^*(n,k^*(n)),k^*(n),R_0^*(n,k^*(n))).&&
    \end{flalign*}
    \end{enumerate}

Now, we write $R_0^*(n^*,k^*(n^*))=R_0^*$, $\boldsymbol{\tau}^*(n^*,k^*(n^*))=\boldsymbol{\tau}^*$ and ${k}^*(n^*)=k^*$. Therefore, $(\boldsymbol{\zeta}^*,R_0^*)$ is the optimal BSP when the form of the decision function is known. \\

    \section{BSP based on Bayes decision function}\label{bayesa}
Here, we consider that the decision criteria is unknown to the manufacturer. Therefore, the decision function is arbitrary. It is also assumed that the manufacturer and consumer use a single prior assessment. An optimal decision function is derived that minimizes the Bayes risk among all the decision functions. This optimal function is known as a Bayes decision
function. In  (\ref{R}), it is seen that the decision function $a(\cdot\ | \ \boldsymbol{\zeta})$ is involved only in the term $R_1(\boldsymbol{\zeta},a)$ of the Bayes risk $R(\boldsymbol{\zeta},a)$. Therefore, to find the Bayes decision function, $R_1(\boldsymbol{\zeta},a)$ is minimized wrt $a$ for fixed $\boldsymbol{\zeta}.$ Note that
$R_1(\boldsymbol{\zeta},a)$ in  (\ref{R1}) can be expressed as
\allowdisplaybreaks\begin{align*}
  R_1(\boldsymbol{\zeta},a)  =&E_{\boldsymbol{\nu}}E_{\boldsymbol{d}\ | \ \boldsymbol{\nu}} [a(\boldsymbol{d}\ | \ \boldsymbol{\zeta})(h(\boldsymbol{\nu})-C_r)]\\
   =&E_{\boldsymbol{d}}E_{\boldsymbol{\nu}\ | \ \boldsymbol{d}} [a(\boldsymbol{d}\ | \ \boldsymbol{\zeta})(h(\boldsymbol{\nu})-C_r)]\\
=&\sum_{\boldsymbol{d}\in\mathcal{X}}[a(\boldsymbol{d}\ | \ \boldsymbol{\zeta})]E_{\boldsymbol{\nu}\ | \ \boldsymbol{d}}[h(\boldsymbol{\nu})-C_r]~P_{\boldsymbol{D}}{( (\boldsymbol{d})}\\
=&\sum_{\boldsymbol{d}\in \mathcal{X}}[a(\boldsymbol{d}\ | \ \boldsymbol{\zeta})]\left\{\int_{\boldsymbol{\nu}}h(\boldsymbol{\nu})~p_{\boldsymbol{\nu}\ | \ \boldsymbol{D}}( \boldsymbol{\nu}\ | \  \boldsymbol{d}) ~d\boldsymbol{\nu}\right\}P_{\boldsymbol{D}}(\boldsymbol{d})\\
=&\sum_{\boldsymbol{d}\in \mathcal{X}}[a(\boldsymbol{d}\ | \ \boldsymbol{\zeta})][\varphi(\boldsymbol{d})-C_r]~P_{\boldsymbol{D}}(\boldsymbol{d}),
\end{align*}
where $p_{\boldsymbol{\nu}\ | \ \boldsymbol{D}}( \boldsymbol{\nu}\ | \  \boldsymbol{d})$ is the posterior JPDF of $\boldsymbol{\nu}$ and it is given by \begin{align*}
    p_{\boldsymbol{\nu}\ | \ \boldsymbol{D}}( \boldsymbol{\nu}\ | \  \boldsymbol{d})=\frac{L(\boldsymbol{\nu}\ | \ \boldsymbol{d})~p(\boldsymbol{\nu})}{\int_{\boldsymbol{\nu}}L(\boldsymbol{\nu}\ | \ \boldsymbol{d})~p(\boldsymbol{\nu})~d\boldsymbol{\nu}},
\end{align*}
$\varphi(\boldsymbol{d})=\int_{\boldsymbol{\nu}}h(\boldsymbol{\nu})~p_{\boldsymbol{\nu}\ | \ \boldsymbol{D}}( \boldsymbol{\nu}\ | \  \boldsymbol{d}) ~d\boldsymbol{\nu}$ and $P_{\boldsymbol{D}}(\boldsymbol{d})=\int_{\boldsymbol{\nu}}~P_{\boldsymbol{D}}(\boldsymbol{d}\ | \ \boldsymbol{\nu})~p(\boldsymbol{\nu})~d\boldsymbol{\nu}$.\\
Now we consider two cases:\\
\textbf{Case 1: } $\varphi(\boldsymbol{d})\leq C_r$:\\
If $a(\boldsymbol{d}\ | \ \boldsymbol{\zeta})=1$, then $ E_{\boldsymbol{d}}E_{ \boldsymbol{\nu}\ | \ \boldsymbol{d}} [a(\boldsymbol{d}\ | \ \boldsymbol{\zeta})(h(\boldsymbol{\nu})-C_r)]\leq 0$ and if $a(\boldsymbol{d}\ | \ \boldsymbol{\zeta})=0$, then $ E_{\boldsymbol{d}}E_{ \boldsymbol{\nu}\ | \ \boldsymbol{d}} [a(\boldsymbol{d}\ | \ \boldsymbol{\zeta})(h(\boldsymbol{\nu})-C_r)]=0$\\
\textbf{Case 2: }$\varphi(\boldsymbol{d})> C_r$:\\
If $a(\boldsymbol{d}\ | \ \boldsymbol{\zeta})=1$, then $ E_{\boldsymbol{d}}E_{ \boldsymbol{\nu}\ | \ \boldsymbol{d}} [a(\boldsymbol{d}\ | \ \boldsymbol{\zeta})(h(\boldsymbol{\nu})-C_r)]\geq0$ and if $a(\boldsymbol{d}\ | \ \boldsymbol{\zeta})=0$, then $ E_{\boldsymbol{d}}E_{ \boldsymbol{\nu}\ | \ \boldsymbol{d}} [a(\boldsymbol{d}\ | \ \boldsymbol{\zeta})(h(\boldsymbol{\nu})-C_r)]=0$\\
Therefore, for each fixed $\boldsymbol{\zeta}$, if we take $ \varphi(\boldsymbol{d})-C_r\leq 0$ when $a(\boldsymbol{d}\ | \ \boldsymbol{\zeta})=1$ and $ \varphi(\boldsymbol{d})-C_r$ $\geq 0$ when $a(\boldsymbol{d}\ | \ \boldsymbol{\zeta})=0$,
then $E_{\boldsymbol{\nu}}E_{\boldsymbol{d}\ | \ \boldsymbol{\nu}} [a(\boldsymbol{d}\ | \ \boldsymbol{\zeta})(h(\boldsymbol{\nu})-C_r)]$ is minimized wrt $a(\cdot\ | \ \boldsymbol{\zeta})$.
Therefore, for fixed $\boldsymbol{\zeta}$, the Bayes decision function $a^*(\boldsymbol{d}\ | \ \boldsymbol{\zeta})$ is given by,
\begin{align*}
    a^*(\boldsymbol{d}\ | \ \boldsymbol{\zeta})=\begin{cases}
        1 &\text{ if } C_r-\varphi(\boldsymbol{d})\geq 0\\
        0 &\text{ otherwise}.
    \end{cases}
\end{align*} 
For no sampling case, the designing parameter $\boldsymbol{\zeta}$ and observed data $\boldsymbol{d}$ are taken as $\boldsymbol{0}$. The prior expectation of  $h(\boldsymbol{\nu})$ is 
\begin{align*} E[h(\boldsymbol{\nu})]=~
   &C_0+\sum_{p=1}^J C_p\frac{\alpha(\alpha_p+1)}{\eta\alpha_0} +\frac{\alpha(\alpha+1)}{\eta^{2}\alpha_0(\alpha_0+1)}\left[\underset{p< q}{\sum_{p=1}^J \sum_{q=1}^J}C_{pq}{(\alpha_p+1)(\alpha_q+1)}+\sum_{p=1}^J C_{pp}(\alpha_p+2)\right]\\
   =~&\varphi(\boldsymbol{0}).
\end{align*}
Therefore, the Bayes decision for the no sampling case can be expressed as
\begin{align*}
    a^*(\boldsymbol{0}\ | \ \boldsymbol{0})=\begin{cases}
        1 &\text{ if } C_r-\varphi(\boldsymbol{0})\geq 0\\
        0 &\text{ otherwise}.
    \end{cases}
\end{align*}
\begin{result}\label{r3}
    If a sample of size $n$ is put on the life test and the items are inspected at an equal length of the interval $h$, then the posterior expectation of the loss function $h(\boldsymbol{\nu})$ can be expressed as
\begin{align*}
    \phi(\boldsymbol{d})=&C_0+N_C\sum_{p=1}^J C_p \left[\sum_{i=0}^{d_t}(-1)^i\binom{d_t}{i}\frac{\Gamma(\alpha+1)}{[(i+s)h+\eta]^{\alpha+1}}\prod_{j=1}^J\frac{\Gamma(\alpha_j+d_{+j}+\delta_{pj})}{\Gamma(\alpha_0+d_t+1)}\right]\\
    &+N_C\underset{p\leq q}{\sum_{p=1}^J \sum_{q=1}^J}C_{pq}\left[\sum_{i=0}^{d_t}(-1)^i\binom{d_t}{i}\frac{\Gamma(\alpha+2)}{[(i+s)h+\eta]^{\alpha+2}}\prod_{j=1}^J\frac{\Gamma(\alpha_j+d_{+j}+\delta_{pj}+\delta_{qj})}{\Gamma(\alpha_0+d_t+2)}\right],
\end{align*}
where  \begin{align*}
     N_C=\left[\sum_{i=0}^{d_t}(-1)^i\binom{d_t}{i}\frac{\Gamma(\alpha)}{[(i+s)h+\eta]^{\alpha}}\prod_{j=1}^J\frac{\Gamma(\alpha_j+d_{+j})}{\Gamma(\alpha_0+d_t)}\right]^{-1}.
 \end{align*}
\end{result}
\begin{proof}
  \noindent Using  (\ref{like}),  when the inspection intervals are equal, the likelihood function can be expressed as
\allowdisplaybreaks\begin{align}
    L(\boldsymbol{\nu}  \ | \ \boldsymbol{d})
&\propto\prod_{m=1}^k\prod_{j=1}^J\left\{\exp\left[-\nu d_{mj}(m-1)h\right]\left[1-\exp(-\nu h)\right]^{d_{mj}}\left(\frac{\nu_j}{\nu}\right)^{d_{mj}}\right\}\exp\left[-\nu(n-d_t)kh\right]\nonumber\\
&\propto\sum_{i=0}^{d_t}(-1)^i \binom{d_t}{i}\prod_{j=1}^J\left(\frac{\nu_j}{\nu}\right)^{d_{+j}}\exp\left[-{(i+s)\nu_j h}\right],
\end{align}
where $s=nk-\sum_{m=1}^k\sum_{j=1}^J(k-(m-1))d_{mj}$.
The prior of $\boldsymbol{\nu}$ follows Gamma-Dirichlet distribution with JPDF 
\begin{align*}
    p(\boldsymbol{\nu})= \frac{\eta^\alpha}{\Gamma(\alpha)}\nu^{\alpha-\alpha_0}\exp(-\eta \nu)\prod_{j=1}^J
    \frac{\Gamma(\alpha_0)}{\Gamma(\alpha_j)}\nu_j^{\alpha_j-1}.
\end{align*}
The posterior JPDF of $\boldsymbol{\nu}$ is
\begin{align*}
    p(\boldsymbol{\nu}\ | \ \boldsymbol{d})= ~N_C
   \sum_{i=0}^{d_t}(-1)^i\binom{d_t}{i}\nu^{\alpha-(\alpha_0+d_t)}\prod_{j=1}^J\nu_j^{\alpha_j+d_{+j}-1}\exp\left[-(\eta+(i+s)h)\nu_j\right],
\end{align*}
where $N_C$ is the normalizing constant of the posterior distribution, which is given by
 \begin{align*}
     N_C=\left[\sum_{i=0}^{d_t}(-1)^i\binom{d_t}{i}\frac{\Gamma(\alpha)}{[(i+s)h+\eta]^{\alpha}}\prod_{j=1}^J\frac{\Gamma(\alpha_j+d_{+j})}{\Gamma(\alpha_0+d_t)}\right]^{-1}.
 \end{align*}
Using Result \ref{r1}, we get the desired result.
\end{proof}\\
For the sampling case, when the length of inspection intervals are unequal, the posterior expectation of $h(\boldsymbol{\nu})$ cannot be obtained analytically. A simple Monte Carlo integration can be used to evaluate it. For this,  a large number $N_2$ of observations are generated from the prior distributions $p(\boldsymbol{\nu})$. Let $\boldsymbol{\nu}^{(1)}\ldots,\boldsymbol{\nu}^{(N_2)}$ be the generated observations. Then 
\begin{align}\label{125}
  \varphi(\boldsymbol{d})\approx  \frac{ \sum_{j=1}^{N_2}h\left(\boldsymbol{\nu}^{(j)}\right)L\left(\boldsymbol{\nu}^{(j)}\ |\ \boldsymbol{d}\right)}{\sum_{j=1}^{N_2}L\left(\boldsymbol{\nu}^{(j)}\ |\ \boldsymbol{d}\right)}
\end{align}
After finding the Bayes decision function, we obtain Bayes risk and the procedure of finding Bayes risk is similar to Section 4. Next, we provide an algorithm for finding optimal BSP under this scenario. \\
\textbf{Algorithm 3}
\begin{enumerate}
    \item Using the derivation of the Bayes decision function discussed in Section \ref{bayesa}, we find the Bayes decision function $a^*(\cdot\ | \ \boldsymbol{\zeta})$ that minimized $R(\boldsymbol{\zeta})$ among all class of decision functions  $a(\cdot\ | \ \boldsymbol{\zeta})$ for each $\boldsymbol{\zeta}=(n, \boldsymbol{\tau}, k)$. 
    \item For each pair of values of $(n,k(n))$,  (\ref{R}) is minimized wrt $\boldsymbol{\tau}$. The procedure is similar to step 2 of Algorithm 2. that is, 
    \begin{flalign*}
       &R(n,\boldsymbol{\tau}^*(n,k(n)),k,a^*(\cdot \ | \ (n,\boldsymbol{\tau}^*(n,k),k))) =\min_{\boldsymbol{\tau}:~\tau_1< \cdots,\tau_k} R(n,\boldsymbol{\tau},k,a^*(\cdot \ | \ (n,\boldsymbol{\tau},k))). &&
    \end{flalign*}
    \item  By minimizing  (\ref{R}) wrt $k$, we find an integer $k^*(n)$, $0\leq k^*(n)\leq k_0(n)$, for each value of $n$, that is,
    \begin{flalign*}
       & R(n,\boldsymbol{\tau}(n,k^*(n)),k^*(n),a^*(\cdot \ | \ (n,\boldsymbol{\tau}^*(n,k^*(n)),k^*(n))))&&\\
        &~~~~=\min_{0\leq m\leq n} R(n,\boldsymbol{\tau}^*(n,k(n)),k(n),a^*(\cdot \ | \ (n,\boldsymbol{\tau}^*(n,k(n)),k))).&&
    \end{flalign*}
    \item Finally, by minimizing  (\ref{R}), we find $n^*$, $0\leq n^*\leq n_0$, that is,
    \begin{flalign*}
       & R(n^*,\boldsymbol{\tau}^*(n^*,k^*(n^*)),k^*(n^*),a^*(\cdot \ | \ (n^*,\boldsymbol{\tau}^*(n^*,k^*(n^*)),k^*(n^*))))&&\\
       &~~~~ =\min_{0<n\leq n_0} R(n,\boldsymbol{\tau}^*(n,k^*(n)),k^*(n),a(\cdot \ | \ (n,\boldsymbol{\tau}^*(n,k^*(n)),k^*(n)))).&&
    \end{flalign*}
    \end{enumerate}
    Now, we write $a^*(\cdot \ | \ (n^*, \boldsymbol{\tau}^*(n^*,k^*(n^*)),k^*(n^*)))=a^*$, $\boldsymbol{\tau}^*(n^*,k^*(n^*))=\boldsymbol{\tau}^*$ and $\boldsymbol{k}^*(n^*)=k^*$. Therefore,  $(\boldsymbol{\zeta}^*,a^*)=(n^*,\boldsymbol{\tau}^*,k^*,a^*)$ is the optimal BSP when the decision function is arbitrary. \\
       \begin{theorem}
        The BSP $(\boldsymbol{\zeta}^*,a^*)$ is the optimal BSP among all BSPs.
    \end{theorem}
    \begin{proof}
        It is enough to prove that for any BSP $\boldsymbol{\zeta}$, the following inequality holds:
        \begin{align*}   R(\boldsymbol{\zeta},a)\geq R(\boldsymbol{\zeta}^*,a^*).
        \end{align*}
        Now,
     \begin{align*} R(\boldsymbol{\zeta},a)-R(\boldsymbol{\zeta}^*,a^*)=&[R(n,\boldsymbol{\tau},m,a)-R(n,\boldsymbol{\tau},m,a^*)]+[R(n,\boldsymbol{\tau},m,a^*)-R(n,\boldsymbol{\tau}^*,m,a^*)]\\
            &+[R(n,\boldsymbol{\tau}^*,m,a^*)-R(n,\boldsymbol{\tau}^*,k^*,a^*)]+[R(n,\boldsymbol{\tau}^*,k^*,a^*)-R(n^*,\boldsymbol{\tau}^*,k^*,a^*)].
        \end{align*}
        From Algorithm 3, we get   
  $[R(n,\boldsymbol{\tau},m,a)-R(n,\boldsymbol{\tau},m,a^*)]\geq 0$, $[R(n,\boldsymbol{\tau},m,a^*)-R(n,\boldsymbol{\tau}^*,m,a^*)]$ $\geq 0$, $[R(n,\boldsymbol{\tau}^*,m,a^*)-R(n,\boldsymbol{\tau}^*,k^*,a^*)]\geq 0$ and $[R(n,\boldsymbol{\tau}^*,k^*,a^*)-R(n^*,\boldsymbol{\tau}^*,k^*,a^*)]\geq 0$.
  Therefore, $[R(\boldsymbol{\zeta},a)\geq R(\boldsymbol{\zeta}^*,a^*)]\geq 0$. Hence, the proof is complete.
    \end{proof}
\section{Numerical Example}\label{num}
Here we provide some examples of determining optimum BSPs. We consider two causes of failure, that is, $J=2$.\\
\textbf{Example 1:} We consider the hyperparameter values of prior distribution as $\alpha=2.8$, $\eta=1$, $\alpha_1=1.5$ and $\alpha_2=1.8$. The  acceptance cost coefficient are taken as $C_0=2$, $C_1=C_2=4$ and $C_{11}=C_{12}=C_{22}=4$. The other cost coefficients are taken as $C_r=40$, $C_s=0.5$, $r_s=0.25$, $C_{\tau}=0.3$ and $C_I=0.1$. $\tau_0=0.1$. Here we denote the BSP as type-I when the decision function is arbitrary and it as type-II when the form of the decision function is known. The optimal sampling schemes corresponding to type I and II are provided in Table \ref{optimal}. Also, in Table \ref{optimal}, we provide the probability of accepting the lot for the optimal BSP, which is denoted by $P(A)$ and is given by
\begin{align*}
    P(A)= \sum_{\boldsymbol{d}\in\mathcal{X}}\int_{\boldsymbol{\nu}} P_{\boldsymbol{D}}(\boldsymbol{d}\ | \ \boldsymbol{\nu})p(\boldsymbol{\nu})~ d\boldsymbol{\nu}.
\end{align*}
 Bayes risk, $E[\tau^*]$, $E[M^*]$ and $E[D_t^*]$  for the optimal BSP are provided in Table \ref{optimal}.
\begin{table}[hbt!]
    \centering
      \caption{Optimal BSP }
    \begin{tabular}{|c|ccccccc|}
    \hline
  Type& $\boldsymbol{\zeta}^*=(n^*,h^*,k^*)$&$R_0^*$&$P(A)$&$E[D_t^*]$&$E[\tau^*]$&$E[I^*]$& Bayes Risk\\
    \hline
     I& (4, 0.30,3)&- &0.489& 3.337& 0.740& 2.467&33.90826\\
    II& (4, 0.30,  3) &0.76& 0.489& 3.337& 0.740& 2.467&33.90826\\
\hline
    \end{tabular}
    \label{optimal}
\end{table}
\subsection{Effect of the parameters}
Here, we study the effect of parameters on designing variables of optimal BSP. The effect of the parameters $C_b$, $C_I$, $C_{\tau}$, $\eta$ and $\alpha$ are tabulated in Tables \ref{cb}, \ref{ck}, \ref{ct}, \ref{beta} and \ref{alpha}, respectively.
\begin{table}[hbt!]
    \centering
      \caption{{Optimal BSPs for  different values of $C_b$}}
    \begin{tabular}{|c|cccccccc|}
    \hline
  $C_b$&Type& $\boldsymbol{\zeta}^*=(n^*,h^*,k^*)$&$R_0^*$&$P(A)$&$E[D_t^*]$&$E[\tau^*]$&$E[I^*]$& Bayes Risk\\
 %   \hline
 %  \multirow{3}{*}{0}     &0&(75,  0.74,  0.04,  4) &  0.564& 94.60& 73.31&  2.55&  3.44&0.06123\\
 %  &0.5&(87, 0.47, 0.91, 15) &  0.600& 94.42& 27.54 &1.06& 2.26& 0.05492\\ 
 %   &1&(85, 0.35,  0.42, 14)  &0.629& 94.86&34.79&  2.33& 6.65& 0.04905\\
    
     \hline 
 \multirow{2}{*}{20}     &I&(0, 0, 0)&-&0&0&0&0& 20\\
 &II&(0, 0, 0)&1&0&0&0&0&20 \\
 \hline
    \multirow{2}{*}{30}     &I&(3, 0.41, 3)&-&0.473 &2.682& 0.855& 2.085&  28.07\\
    &II&(3,  0.41,  3)&0.80& 0.333& 2.682 &0.855& 2.085& 28.07 \\
       \hline
 \multirow{2}{*}{50}  &I& (6, 0.27,  2)&-&0.528 &4.209 &0.525& 1.944& 38.21502\\
 &II& (6,  0.27,  2) & 0.72& 0.620& 4.209 &0.525& 1.944&38.21517 \\
   
   \hline
 
   \multirow{2}{*}{60}     &I&(7,  0.21,  2)&-&0.529 &4.378 &0.416 &1.982&41.44479 \\
  &II&(7,  0.22,  2) &  0.68& 0.714& 4.478& 0.435& 1.978&41.44624 \\
   \hline
  \multirow{2}{*}{70} &I&(3,  0.19, 2)&-&0.726&2.98&0.38&1.96& 44.69447\\
&II&(3, 0.19, 2)&0.52&0.726&2.98&0.38&1.96& 44.69447 \\
 \hline 
    \multirow{2}{*}{90} &I&(0, 0, 0)&-&1&&&& 47.6619\\
 &II&(0,0,0)&0&1&0&0&0&47.6619 \\
   \hline 
    \end{tabular}
    \label{cb}
\end{table}\\
  In Table \ref{cb}, it is seen that the optimal solution is $(0,0,0)$ when $C_b=20$ and $C_b=90$. This is the no sampling case. For type-I BSP, when $C_b=20$,  the Bayes risk is equal to $C_r$, this means that the lot is rejected without life testing. When $C_b=90$, the Bayes risk is equal to $\varphi(0,0,0)$, which means that the lot is accepted without life testing.  In Table \ref{cb}, it is seen that the sample size increases up to a certain value and after that it decreases, when $C_b$ increases. This is due to the fact that the distance between $C_b$ and prior expectation $\varphi(\boldsymbol{0})$ is decreasing with $C_b$ on $(-\infty, 47.66)$ and increasing with $C_b$ on $(47.66,\infty)$. Also, it is seen that the probability of acceptance and Bayes risk increases with $C_b$ as expected.
\begin{table}[hbt!]
    \centering
      \caption{Optimal BSPs for different values of $C_I$ }
    \begin{tabular}{|c|cccccccc|}
    \hline
  $C_I$&$Type$& $\boldsymbol{\zeta}^*=(n^*,h^*,m^*)$&$R_0^*$&$P(A)$&$E[D_t^*]$&$E[\tau^*]$&$E[I^*]$& Bayes Risk\\
    \hline 
 \multirow{3}{*}{0}     &I&(5, 0.15, 6)&-& 0.510& 4.171& 0.732& 4.878& 33.57594\\
 &II& (5, 0.16, 5)&0.76& 0.515& 4.036& 0.680& 4.252&33.63195 \\
   \hline
 
   \multirow{3}{*}{0.2}     &I&(5, 0.47, 1)&-& 0.50 &3.30& 0.47 &1.00&  34.09155 \\
  &II& (5,  0.47,  1) &0.72& 0.50& 3.30& 0.47& 1.00&34.09155 \\
   \hline
   \multirow{3}{*}{0.3} &I&(5, 0.47, 1)&-& 0.50 &3.30& 0.47 &1.00&  34.19155 \\
 &II&(5, 0.47,  1)  &0.72& 0.50& 3.30& 0.47& 1.00& 34.19155 \\
   \hline 
    \multirow{3}{*}{1} &I& (5,  0.47,  1)&-& 0.50 &3.30& 0.47 &1.00&   34.89155\\
 &II&(5, 0.47,  1)  &0.72& 0.50& 3.30& 0.47& 1.00&34.89155 \\
   \hline 
    \end{tabular}
    \label{ck}
\end{table}
\begin{table}[hbt!]
    \centering
      \caption{Optimal BSPs for different values of $C_{\tau}$ }
    \begin{tabular}{|c|cccccccc|}
    \hline
  $C_{\tau}$& Type&$\boldsymbol{\zeta}^*=(n^*,h^*,m^*)$&$R_0^*$&$P(A)$&$E[D_t^*]$&$E[\tau^*]$&$E[I^*]$& Bayes Risk\\
    \hline  
        \hline 
 \multirow{3}{*}{0}&I& (4, 0.36,  3) & -&0.486& 3.485& 0.838 &2.329&33.66994 \\
 &II&(4,  0.36,  3)&0.74 &0.486& 3.485& 0.838& 2.329& 33.66994 \\
 \hline
  \multirow{3}{*}{0.1}     &I&(4, 0.36,  3)&- &0.486& 3.485& 0.838 &2.329& 33.75377\\
    &II& (4,  0.36,  3)&0.74 &0.486& 3.485& 0.838& 2.329& 33.75377 \\
       \hline
 \multirow{3}{*}{0.5}  &I&  (5, 0.29,  2)&- &0.502& 3.611& 0.553& 1.908&34.04424 \\
 &II& (5, 0.29, 2) &0.76& 0.502& 3.611 &0.553& 1.908&34.04424 \\
   
   \hline
 
   \multirow{3}{*}{1}     &I&(5, 0.25,  2) &-&0.500& 3.393 &0.484& 1.936&34.31483 \\
  &II& (5, 0.25,  2) &0.76&0.500& 3.393 &0.484& 1.936&34.31483 \\
      \hline
    \end{tabular}
    \label{ct}
\end{table}
\begin{table}[hbt!]
    \centering
      \caption{Optimal BSPs for different values of $\eta$ }
    \begin{tabular}{|c|cccccccc|}
    \hline
  $\eta$& Type&$\boldsymbol{\zeta}^*=(n^*,h^*,m^*)$&$R_0^*$&$P(A)$&$E[D_t^*]$&$E[\tau^*]$&$E[I^*]$& Bayes Risk\\
  \hline
   \multirow{3}{*}{0.5}     &I& (0, 0, 0)&- &0&0& 0& 0&  40\\
   &II& (0, 0, 0)&1&0&0& 0& 0& 40\\
    \hline
   \multirow{3}{*}{0.75}     &II& (3, 0.37,  3)&-& 0.261& 2.764 &0.716& 1.934& 37.74172 \\
   &II& (3, 0.37,  3) &  0.79& 0.261& 2.764 &0.716& 1.934& 37.74172\\
   \hline
 \multirow{3}{*}{1.25}     &I& (5, 0.38,  1)&-&0.703& 2.622& 0.380& 1.000& 29.46731\\
 &II& (5, 0.38,  1)&  0.66  &0.703& 2.622& 0.380 &1.000&29.46731\\
 \hline
  \multirow{3}{*}{1.5}     &I&(0, 0, 0)&-&1&0&0&0& 24.78307\\
 &II&  (0, 0, 0)& 0 &1&0&0&0&24.78307\\
       \hline

    \end{tabular}
    \label{beta}
\end{table}
\begin{table}[hbt!]
    \centering
      \caption{Optimal BSPs for different values of $\alpha$ }
    \begin{tabular}{|c|cccccccc|}
    \hline
  $\alpha$& Type&$\boldsymbol{\zeta}^*=(n^*,h^*,m^*)$&$R_0^*$&$P(A)$&$E[D_t^*]$&$E[\tau^*]$&$E[I^*]$& Bayes Risk\\
  \hline
   \multirow{3}{*}{2.2}     &I&(5, 0.25,  2) &- &0.690& 2.951& 0.490& 1.961&28.28359 \\
   &II& (5, 0.40, 1)&0.670&0.691 &2.615& 0.400& 1.000& 28.288779\\
    \hline
   \multirow{3}{*}{2.5}     &I& (5, 0.27,  3)&-&0.599& 3.866& 0.727& 2.691& 31.31123 \\
   &II& (5,  0.43,  1)&0.690&0.601& 2.955 &0.430 &1.000& 32.25544\\
   \hline
 \multirow{3}{*}{3}     &I&(4, 0.31,  3)&- &0.434 &3.444 &0.741 &2.390&35.38329\\
 &II& (5, 0.30,  2)&0.77& 0.446 &3.779& 0.566& 1.887&35.45693\\
 \hline
  \multirow{3}{*}{3.3}     &I&(4, 0.40,  3) &-&0.343& 3.703 &0.836& 2.090&37.22465\\
 &II& (5, 0.32,  2)&0.78 &0.360 &4.023 &0.590 &1.844&37.36195\\
       \hline

    \end{tabular}
    \label{alpha}
\end{table}\\
 In Table \ref{ct}, it is observed that the expected duration of life test decreases when $C_{\tau}$ increases as expected. In Table \ref{alpha}, we find that for type-II BSP, the sample size $n$ remains unchanged for the different values of $\alpha$. However, the time interval $h$ and inspection number change with $\alpha$. In Table \ref{beta}, a similar fact is seen as in Table \ref{cb}.
Tables 1-6 show that in most cases, optimum BSPs do not change with the types of BSP. 
\subsection{Illustrative example using approximate Bayes risk}
\textbf{Example 2:}
We consider the hyperparameter values of prior distribution as $\alpha=2.8$, $\eta=1$, $\alpha_1=1.5$ and $\alpha_2=1.8$. The  acceptance cost coefficient are taken as $C_0=2$, $C_1=C_2=4$ and $C_{11}=C_{12}=C_{22}=4$. The other cost coefficients are taken as $C_r=40$, $C_s=0.15$, $r_s=0.1$, $C_{\tau}=0.3$, and $C_I=0$ and $\tau_0=0.1$. The optimal BSP is provided in Table \ref{optimal1}. 
\begin{table}[hbt!]
    \centering
      \caption{Optimal BSP using approximate Bayes Risk function }
    \begin{tabular}{|ccccccc|}
    \hline
   $\boldsymbol{\zeta}^*=(n^*,h^*,k^*)$&$R_0^*$&$P(A)$&$E[D_t^*]$&$E[\tau^*]$&$E[I^*]$& Bayes Risk\\
    \hline
     (13, 0.102, 5)&0.761 &0.489& 3.337& 0.740& 2.467& 31.4662\\
\hline
    \end{tabular}
    \label{optimal1}
\end{table}
\begin{table}[hbt!]
 \centering
      \caption{Optimal BSP for different value of $C_s$ and $C_I$ using approximate Bayes Risk function}
    \begin{tabular}{|cc|ccccccc|}
    \hline
 $C_s$& $C_I$&  $\boldsymbol{\zeta}^*=(n^*,h^*,k^*)$&$R_0^*$&$P(A)$&$E[D_t^*]$&$E[\tau^*]$&$E[I^*]$& Bayes Risk\\
    \hline
    0.12& 0.01& (17,  0.100,  4)&  0.759&0.576&10.373 & 0.399 & 3.990& 31.0691087\\
 0.12& 0.05& (19,  0.159,  2)&  0.759&0.575&10.230 & 0.318 & 2.000& 31.1659256\\
0.12&0.1&(19, 0.159, 2)&0.759&0.575&10.230&  0.318 & 2.000& 31.26591\\
\hline
    0.15&0.01&(13, 0.103,  5)& 0.761&0.573&8.938& 0.507& 4.924& 31.5154582\\
    0.15&0.05&(14, 0.149,  3)&0.761&0.572&9.025 &0.445 &2.986& 31.6515661\\
 0.15& 0.1& 
(15,  0.194, 2)&0.761&0.570&9.010 &0.388& 1.999& 31.7571885\\
\hline
0.18& 0.01& (11, 0.117,  5)&0.763&0.569&7.971 &0.568 &4.850&31.8755930\\
0.18& 0.05& (12, 0.166, 3)& 0.763&0.568&8.130 &0.493 &2.969& 32.0293011\\
0.18& 0.1& (12, 0.216,  2)&0.763&0.567&7.609& 0.431& 1.995&32.1584409\\
\hline
    \end{tabular}
    \label{optimal2}
\end{table}
For different values of $C_I$ and $C_s$, the optimal BSPs using approximate Bayes Risk function are tabulated in Table \ref{optimal2}. It is seen that when $C_I$ is fixed, the sample size decreases with $C_s$ and when $C_s$ is fixed the sample size increases with $C_I$. Also, the acceptance limit $R_0^*$ increases with $C_s$. 
\subsection{Illustration of the decision function after life testing } In all cases, after determining the value of $\boldsymbol{\zeta}$, the manufacturer conducts the life testing to take a decision about the lot. In Example 1, the optimal plan is $(4, 0.30, 3)$ for the equal inspection intervals. Therefore, in that case, the manufacturer conducts a life test with a sample size of $4$ and after a time length of $0.30$, the number of failures due to each cause is observed. The decision about the lot is taken based on that data. In that scenario, the acceptance probability and optimal sampling parameter both are same. Therefore, we can say that the acceptance set for both types of BSPs are same. Now, we generated six observations using the optimal design $(4,0.30,3)$ which are tabulated in Table \ref{generated data}. For type-I BSP, the decision is taken for the observed data $\boldsymbol{d}$ based on the function $\varphi(\boldsymbol{d})$ and for type-II BSP, the decision is taken based on estimated reliability. In Table \ref{generated data}, both the functions are calculated for the observed data $\boldsymbol{d}$. Lastly, the decision for the observed data $\boldsymbol{d}$ are tabulated in Table \ref{generated data}. For type-I BSP, the lot is accepted when $\varphi(\boldsymbol{d})< C_r$ and for type-II BSP, the lot is accepted when the estimated reliability is greater than 0.76.
\begin{table}[hbt!]
     \centering
     \caption{The data sets, estimated reliability and corresponding decision about the lot}
  \small   \begin{tabular}{|c|ccc|ccc|}
     \hline
  $i$&\multicolumn{3}{c|}{ Data} & Estimated & $\varphi(\boldsymbol{d})$&Decision \\
        \cline{2-4} &$(d_{11},d_{12})$&$(d_{21},d_{22})$&$(d_{31},d_{32})$&Reliability&&\\
         \hline
         1&(0, 0)& (0, 0)& (0, 1)& 0.971&8.439& Accept the lot\\
      2&  (0, 0)&(0, 0)&(0, 0)&0.973&6.063&Accept the lot\\
      3& (0, 0)& (1, 0)& (2, 1)& 0.860&22.219&Accept the lot\\
      4 & (1, 1)& (1, 0)& (1, 0)&0.753&42.531& Reject the lot\\
      5&(2, 0)&(2, 0)&-& 0.693&55.898&Reject the lot\\
     6&(2, 1)& (0, 0)& (0, 1) &0.693&53.155& Reject the lot\\
      \hline 
     \end{tabular}
     \label{generated data}
\end{table}

In Table \ref{generated data}, it is seen that for $2^{nd}$ observed data, all $4$ items failed after the $2^{nd}$ inspection. Therefore, no samples are surviving for $3^{rd}$ inspection. For that reason, $E[M]$ is less than $k$ and $E[\tau]$ is less than $\tau_k$. 
\section{Conclusion}\label{con}
In this work, we have considered determining the optimal BSP under ICS for competing risk data. The advantage of using ICS is that continuous monitoring is not required. The work can be extended to progressive type-I ICS. We have considered an exponential distribution for illustration. In this work, it is observed that for ICS, the optimal BSP based on reliability acceptance criteria is similar to the optimal BSP based on the Bayes decision function. Therefore, for a large sample, using asymptotic properties of MLEs, the approximate form of Bayes risk can be considered to find an optimal BSP. The proposed methodology of the approximate form of Bayes risk can be extended to other lifetime distributions.\\
Also, in these works, a single prior assessment is used for the parameters of the lifetime distribution. However, the consumer and the manufacturer have different prior knowledge about the parameters of the lifetime distributions. Different prior assessment for the parameters of the lifetime distribution for single failure mode items was studied by \citet{das2024bayesian,das2024hybrid}. Using this concept, the work can be extended to a model where the consumer and the manufacturer have different prior knowledge about the parameters of the lifetime distributions for multiple failure mode items.
%\section*{Acknowledgement}
\bibliographystyle{unsrtnat}
\bibliography{citiation}
% \section*{Author Biographies}
% \begin{longtable}{|c|}
% \hline
% \\
% \begin{minipage}{0.16\textwidth}
%  {\includegraphics[width=1in,height=1.25in]{Biswabrata.jpg}}
% \end{minipage}
% \begin{minipage}{0.73\textwidth}
% \textbf{Biswabrata Pradhan} is a Professor at the Statistical Quality Control \& Operations Research (SQC \& OR) Unit, Indian Statistical Institute, Kolkata. He received his M. Tech. in Quality, Reliability and Operations Research, and PhD in Statistics from the Indian Statistical Institute in 1996 and 2010, respectively.  His primary areas of research include inference based on censored data, design of censored life testing experiments, reliability theory and survival analysis
% \end{minipage}\\\\
% \begin{minipage}{0.16\textwidth}
%     {\includegraphics[width=1in,height=1.25in]{markup_1000125920 (1) (2).jpg}}
% \end{minipage}
% \begin{minipage}{0.73\textwidth}
%     \textbf{Rathin Das} is currently a PhD student in the Statistical Quality Control and Operations Research (SQC \& OR) Unit of the Indian Statistical Institute (ISI), Kolkata, and pursuing his PhD work in Quality, Reliability, and Operations Research (QR \& OR) from ISI. He completed his B.Sc. (Hons.) degree in Mathematics from the University of Burdwan in 2018 and completed his M.Sc. in Applied Mathematics in 2020 from the Indian Institute of Engineering Science and Technology, Shibpur. 
% \end{minipage}
% \\\\
% \hline
% \end{longtable}

\end{document}